\documentclass[a4paper, 12pt, russian]{article}
\usepackage{cmap}			
\usepackage[T2A]{fontenc}	
\usepackage[utf8]{inputenc}
\usepackage[russian]{babel}
\usepackage{amsmath,amssymb} 
\usepackage[unicode, pdftex]{hyperref}
\usepackage{amsfonts}
\usepackage{mathtext} 
\usepackage{scrextend}
\usepackage[usenames]{color}
\usepackage{colortbl}
\usepackage{graphicx} 
\usepackage{amsthm}
\usepackage{fullpage}
\usepackage{cite}
\newtheorem{lemma}{Лемма}
\newtheorem{theorem}{Теорема}
\newtheorem{definition}{Определение}

\newtheorem{proposition}{Предложение}
\newtheorem{remark}{Замечание}

\newcommand{\erf}{\mathrm{erf}}
\DeclareMathOperator{\sign}{sign}
\date{}

\begin{document}
\renewcommand{\refname}{Список литературы}

\title{
Скорость сходимости алгоритмов решения линейного уравнения методом квантового отжига
\footnote{Авторы благодарят Михаила Скопенкова за внимание к работе и полезные замечания. Исследование разделов 1, 2, 3.1 выполнено при поддержке гранта Российского научного фонда №~21-11-00047. Исследование разделов 3.3 и 3.4 выполнено в Санкт-Петербургском международном математическом институте имени Леонарда Эйлера при финансовой поддержке Министерства науки и высшего образования Российской Федерации (соглашение №~075–15–2022–287 от 06.04.2022). Исследование раздела 3.2 выполнено при поддержке Projeto Paz и Coordenação de Aperfeiçoamento de Pessoal de Nível Superior - Brasil (CAPES) - Finance Code 001.}}
\author{В.С. Шалгин$^1$, С.Б. Тихомиров$^2$}

\maketitle

\vspace{-0.5cm}

{\footnotesize
\begin{tabbing}
    $^1$ \= Санкт-Петербургский государственный университет, \+\\Россия, 199034, Санкт-Петербург, Университетская наб. 7/9.
\end{tabbing}
\begin{tabbing}
    $^2$ \=Pontifícia Universidade Católica do Rio de Janeiro - PUC-Rio,\+\\Rua Marquês de São Vicente, 225, Gávea - Rio de Janeiro, \\RJ - Brasil
Cep: 22451-900 - Cx. Postal: 38097
\end{tabbing}}

{\small
\begin{quote}
\noindent{\sc Аннотация.} Мы рассматриваем различные итеративные алгоритмы решения линейного уравнения $ax=b$ при помощи квантового вычислительного устройства, работающего по принципу квантового отжига. В предположении, что результат работы компьютера описывается распределением Больцмана, показано, при каких условиях алгоритмы решения уравнения сходятся, и дана оценка на скорость их сходимости. Рассмотрено применение данного подхода для алгоритмов, использующих как бесконечное количество кубитов, так и малое количество кубитов.


\noindent{\bf Ключевые слова:} адиабатические квантовые вычисления, квантовый отжиг, линейное уравнение, распределение Больцмана, усеченное нормальное распределение.
\end{quote}
}

\section{Введение}

\hyphenation{quan-tum}
\hyphenation{an-neal-er}

Квантовые вычисления представляют собой новую парадигму выполнения вычислений, предложенную Ю. И. Маниным \cite{Manin} и Р. Фейнманом \cite{Feynman}. Функционирование таких вычислительных устройств основано на квантовой механике. В основе вычислений лежат квантовые биты (кубиты), которые могут находиться не только в состоянии ``0'' или ``1'', но и в их суперпозиции. Что более важно, квантовые биты могут находиться в запутанном состоянии. Таким образом, система из $n$ кубитов описывается $2^n$ комплексными числами, более того операция над одним кубитом ``меняет состояние'' всех запутанных кубитов, что является основой квантового параллелизма --- одновременного проведения вплоть до $O(2^n)$ операций над числами, описывающими состояние системы \cite{Williams, Nielsen}. В квантовых вычислениях появляются и ограничения, не свойственные классическим вычислениям, например, невозможность копирования состояния и считывания состояния без его изменения. Для множества задач разработаны алгоритмы для квантовых компьютеров, работающие быстрее, чем их классические аналоги, вплоть до экспоненциального ускорения. Например, алгоритм поиска (алгоритм Гровера) \cite{Grover}, разложение на множители (алгоритм Шора) \cite{ShorQuantumSupremacy}, приближенное решение систем линейных уравнений \cite{HHL}.

Есть две основные модели квантовых вычислений: универсальная схемная модель (circuit based) \cite{Williams, Nielsen} и адиабатическая модель \cite{AdiabaticQuantumComputation}. В схемной модели операции выполняются одна за другой, как в классических вычислениях. Операции --- квантовые вентили --- представляют собой унитарные операторы, действующие на состояние системы кубитов. На основе этой модели функционируют квантовые компьютеры таких компаний как IBM, Google, Intel. Основными ограничениями для практического использования таких квантовых компьютеров в настоящее время является небольшой размер (порядка 100 кубитов) и низкая точность выполнения операций. 

Принцип работы адиабатических квантовых компьютеров основан не на последовательном выполнении операций, а на адиабатической теореме \cite{AdiabaticQuantumComputation}. Если изначально система находилась в состоянии минимальной энергии для гамильтониана $H_1$, то при достаточно медленной эволюции гамильтониана 
$$
    H(t) = (1-vt)H_1+ vtH_2, \quad t \in [0, 1/v]
$$ 
в конечный момент времени система будет находится в состоянии минимальной энергии для гамильтониана $H_2$. Гамильтониан $H_2$ строится таким образом, что состояние минимальной энергии для него будет решением некоторой задачи. Такой подход позволяет эффективно решать задачи дискретной оптимизации, например, задача коммивояжера \cite{TravellingSalesman} и задача разрешимости булевых функций \cite{QuantumComputationbyAdiabaticEvolution}. Известно, что адиабатическая модель эквивалентна универсальной схемной модели квантовых вычислений \cite{AdiabaticIsEquiv}.

С адиабатическими квантовыми вычислениями тесно связан квантовый отжиг (quantum annealing) --- квантовый аналог алгоритма имитации отжига \cite{quantum_annealing_paper}. Устройство, работающее по принципу квантового отжига, носит название ``quantum annealer'' (QA). Процесс поиска также начинается с состояния минимальной энергии системы для гамильтониана $H_1$. Однако структура целевого гамильтониана $H_2$ более ограничена по сравнению с той, что фигурирует в общей адиабатической модели. А именно, квантовый отжиг нацелен на поиск точки минимума целевой функции модели Изинга \cite{quantum_annealing_paper, ising}. В упрощенном виде ее можно представить так:
\begin{equation}\label{Ising_func}
    F(\sigma)=\sum_i h_i\sigma_i+\sum_{i<j} J_{ij}\sigma_i\sigma_j,
\end{equation}
где $\sigma_i\in\{-1,1\}$ представляют собой спины кубитов, а $h_i$ и $J_{ij}$ --- коэффициенты линейных и квадратичных слагаемых соответственно. Результат работы квантового отжига --- набор спинов кубитов $\{\sigma_i\}$, которые доставляют минимум функции $F(\sigma)$.

Точный результат может быть получен только в случае нулевой абсолютной температуры у QA, что в текущих реализациях является недостижимым. На практике такое устройство будет выдавать сэмпл из распределения Больцмана \cite{AlbashTemperatureScaling}. Вероятность получить состояние $\sigma$ зависит от значения функции $F$ и обратной температуры~$\beta$:
\[
    P(\sigma)\propto e^{-\beta F(\sigma)}.
\]
Ввиду этого квантовый отжиг является неточным, эвристическим алгоритмом. Кроме того, наличие шума также может внести помехи в работу компьютера.

Высокий интерес к модели квантового отжига обусловлен наличием реализации устройства, работающего по данной модели с большим количеством кубитов. Соответствующая реализация представлена устройствами компании D-Wave Systems \cite{dwavedocs}, количество кубитов в которых достигает 5000. Дополнительными ограничениями в работе компьютера являются граф связности между кубитами и неточность выполнения операций (квантовый шум). Перед тем, как решить задачу с помощью компьютера D-Wave ее необходимо перевести в термины модели Изинга. Согласованность распределения Больцмана и результата работы D-Wave достаточно хорошо продемонстрирована \cite{ising, Boltz1, Boltz2, Boltz3}. 
Для обсуждения вопроса о наличии превосходства QA над классическими компьютерами см., например, \cite{QuestionQuantum1, Nature5000}.

В нашей работе мы будем опираться на модель квантовых вычислений, работающую по принципу квантового отжига. Одной из важных для приложений задач является решение систем линейных алгебраических уравнений. Задача решения системы $Ax=b$ эквивалентна задаче минимизации функции $||Ax-b||^2$, часто называемой в литературе ``linear least squares problem'' (LLS). Данная задача может быть решена при помощи QA путем ее перевода в целевую функцию модели Изинга. Заметим, что линейность по переменной $x$ является необходимым условием, так как в противном случае целевая функция не будет иметь форму \eqref{Ising_func}. 

Во многих работах исследуется решение задачи LLS с помощью QA. В статье \cite{OMalleyVesselinovLLS} предлагается подход к переформулировке задачи LLS в форму модели Изинга. Также в работе авторы предположили, что QA лучше всего подходит для решения задачи LLS в случае, если матрица $A$ разреженная или когда компоненты вектора $x$ бинарные. Этот подход получил дальнейшее развитие. В статье \cite{BorleLomonaco} предлагается подход к решению произвольной системы линейных уравнений, а также приводятся условия, при которых возможно получить ускорение по сравнению с лучшим из известных классических алгоритмов решения произвольных систем линейных уравнений. В статье \cite{RogersFloatingpoint} рассматривается задача решения одного уравнения с одной неизвестной и задача LSS, подробно излагается процесс переформулировки исходной задачи в форму модели Изинга и встраивания полного графа задачи в граф компьютера D-Wave. В статье \cite{BorleHowViable} исследуется вопрос о целесообразности использования QA для решения систем уравнений и предлагается гибридный алгоритм решения линейных систем. Подход решения линейных систем с помощью QA нашел применение во многих задачах, в которых возникает необходимость решения системы уравнений: оценке линейной регрессии \cite{linearRegression}, для задач сейсмической томографии \cite{seismic}, в задаче определения преобразования из точечного множества \cite{MML2022}, решения краевой задачи для эллиптических уравнений \cite{Elliptic2023}. 

Во всех упомянутых работах были предложены алгоритмы решения линейных уравнений и систем, включая и итеративные алгоритмы \cite{RogersFloatingpoint, seismic, MML2022,Elliptic2023}. Однако рассматривалась только экспериментальная постановка задачи, теоретические вопросы сходимости подобных итеративных алгоритмов не исследовались.

В данной работе рассматривается подход к решению задачи LLS аналогичный \cite{RogersFloatingpoint, seismic, MML2022, Elliptic2023} для случая одного уравнения с одной неизвестной:
\[
    ax=b.
\]
В рамках данного подхода мы учитываем подверженность квантового компьютера ошибкам и вероятностную природу QA. Нас будет интересовать вопрос устойчивости итеративного алгоритма к погрешностям, связанным с вероятностной обусловленностью результата работы компьютера. Ввиду этого простейшее уравнение является подходящей моделью для анализа сходимости и устойчивости итеративных алгоритмов. Так как результат работы QA подчиняется распределению Больцмана, это позволяет производить оценку работы алгоритмов на основе нормальных распределений и их модификаций.

Мы рассматриваем различные итеративные алгоритмы, которые работают как для большого (стремящегося к бесконечности) количества кубитов так и для малого числа кубитов. Характерной чертой рассмотренных алгоритмов является адаптация размера поправки на каждом шаге в зависимости от текущего значения невязки, что аналогично подходам \cite{RogersFloatingpoint, seismic, MML2022,Elliptic2023} и расширяет предложенные ранее подходы \cite{OMalleyVesselinovLLS, BorleLomonaco, linearRegression}. Мы доказываем, что предложенные алгоритмы сходятся к точному значению, при достаточно малых ошибках в квантовом компьютере и оцениваем скорость сходимости (см. теоремы \ref{Th_xn_normal_convergence_general}, \ref{Th_general_model}). Для реализации алгоритмов адаптации мы используем  сложение, умножение и домножение на целые степени двойки, что соответствует сдвигу битов, и не используем деление на произвольное число.

В разделе \ref{section_preliminary} мы даем предварительные определения, переводим задачу решения уравнения в термины модели, эквивалентной модели Изинга, и устанавливаем вероятностную модель вычислений. В разделе \ref{section_improving} мы рассматриваем итеративные алгоритмы, основанные на последовательном улучшении приближенных решений уравнения. В разделе \ref{subsection_improving_solution_normal_distr} мы рассматриваем идеализированный случай, при котором результат работы QA подчиняется закону нормального распределения. В разделе~\ref{subsection_improving_general} мы рассматриваем общий подход к изучению скорости сходимости итеративных алгоритмов. Случай усеченного нормального распределения решений, соответствующий бесконечному количеству кубитов, рассмотрен в разделе \ref{subsection_improving_solution_cut_normal_distr}, случай распределения Больцмана, соответствующий конечному количеству кубитов, --- в разделе~\ref{subsection_improving_solution_Boltzmann_distr}.


\section{Предварительные построения}
\label{section_preliminary}

\hyphenation{un-cons-trained}

Модель Изинга эквивалентна, так называемой, модели QUBO (quadratic unconstrained binary optimization) \cite{qubo}:
\begin{equation}\label{qubo}
    H(q_1,q_2,\ldots,q_n) = \sum_{i=1}^{n}\sum_{j=i}^{n} Q_{ij} q_i q_j,
\end{equation}
где $(Q_{ij})_{i,j=1}^n$ --- верхнетреугольная квадратная матрица порядка $n$, $q_i\in\{0,1\}$. Для переформулировки задачи решения уравнения $ax=b$ в задачу QUBO мы используем функцию 
\begin{equation}\label{hamiltonian_expression}
    H(x) = (ax-b)^2.
\end{equation}
Переменную $x$ представляем с конечной точностью в виде
\begin{equation}\label{expression_for_x}
    x=\vartheta q_{p} +\sum\limits_{i=r}^{p-1} 2^i q_i,
\end{equation}
где $p,r\in\mathbb{Z}$, $r<p$, $\vartheta=-2^p+2^r$,  $q_i\in\{0,1\}$. Бит $q_p$ отвечает за знак переменной $x$. Роль константы $\vartheta$ заключается в том, что набор битов $q_r\ldots q_{p-1}$ для отрицательных значений $x$ представляет собой дополнительный код к аналогичному набору битов для положительных значений $x$. Отметим, что можно использовать и другие представления переменной \cite{BorleLomonaco, RogersFloatingpoint, Elliptic2023}
Обозначим через $\Omega_{r,p}$ множество чисел вида \eqref{expression_for_x}. Тогда
\begin{equation*}\label{Omega_r_p}
    \Omega_{r,p}=\left\{\pm\sum_{i=r}^{p-1} q_i 2^i~:~q_i\in\{0,1\}\right\}.
\end{equation*}
После подстановки \eqref{expression_for_x} в \eqref{hamiltonian_expression} и отбрасывания постоянного слагаемого мы получаем целевую функцию вида \eqref{qubo}:
\begin{equation*}
    H(q_{r},\ldots,q_{p})
    =\sum_{i=r}^{p}\sum_{j=i}^{p}{Q_{ij}q_i q_j},
\end{equation*}
где
\begin{equation*}
    Q_{ii}=
    \begin{cases}
        a^2\vartheta^2-2ab\vartheta,~~~~i=p,\\
        2^{2i}a^2-2^{i+1}ab,~r\leqslant i\leqslant p-1,
    \end{cases}
    Q_{ij}=
    \begin{cases}
        2^{i+1}a^2\vartheta,~~j=p,r\leqslant i\leqslant p-1,\\
        2^{i+j+1}a^2,~r\leqslant i<j\leqslant p-1.
    \end{cases}
\end{equation*}

Как было сказано ранее, ошибки в работе QA имеют вероятностный характер, и мы считаем, что они подчиняются распределению Больцмана. Мы будем использовать вариацию этого распределения, определенную ниже.

\begin{definition}\label{def_Boltzmann}
    Пусть $\Omega$ --- конечное подмножество вещественных чисел, $H: \mathbb{R}\rightarrow[0,+\infty)$, $\beta>0$. Распределением Больцмана $\mathrm{B}\left(\beta,\Omega,H(x)\right)$ на множестве $\Omega$ с параметром $\beta$ и целевой функцией $H(x)$ назовем вероятностное распределение на $\Omega$, в котором вероятность элемента $x\in\Omega$ определяется как
    \[
        P(x)= \frac{1}{Z}e^{-\beta^2 H(x)}, \quad \mbox{где $Z=\sum_{y\in\Omega}e^{-\beta^2 H(y)}$}.
    \]
\end{definition}

Чем меньше значение $H(x)$, тем больше вероятность получить $x$ в качестве результата работы компьютера. Параметр $\beta$ отражает точность работы компьютера. Чем больше значение $\beta$, тем более вероятно результат работы компьютера будет близок к точке минимума функции $H(x)$ на множестве $\Omega$. Если $\beta\rightarrow+\infty$, то компьютер будет работать без ошибок.

Вернемся к целевой функции \eqref{hamiltonian_expression}. Если предположить, что количество кубитов в QA стремится к бесконечности, и для двоичного представления \eqref{expression_for_x} переменной $x$ мы используем как все положительные степени двойки, так и все отрицательные, то распределение решений будет стремиться к нормальному, что показывает следующее очевидное предложение.

\begin{proposition}\label{St_limit_measure}
    Пусть $r,p\in\mathbb{Z}$, $r<p$, $\Omega_{r,p}=\left\{ \pm\sum\limits_{i=r}^{p-1} q_i\,2^i\,:\,q_i\in\{0,1\}\right\}$, $\beta>0$, $a,b\in\mathbb{R}$, $a\neq0$. Тогда $$\lim\limits_{\substack{p\rightarrow+\infty\\ r\rightarrow-\infty}}
    \mathrm{B}\left(\beta,\Omega_{r,p},(ax-b)^2\right)=\mathcal{N}\left(\frac{b}{a},\frac{1}{2a^2\beta^2}\right)$$ по распределению.
\end{proposition}

Таким образом, в качестве приближения к распределению Больцмана, мы можем использовать нормальное распределение.


\section{Улучшение решения уравнения}
\label{section_improving}

\subsection{Модель улучшения решения, основанная на нормальном распределении}\label{subsection_improving_solution_normal_distr}

В этом разделе мы будем предполагать, что результат работы QA по решению уравнения $ax=b$ имеет распределение $\mathcal{N}\left(\frac ba, \frac{1}{2a^2\beta^2}\right)$ согласно предложению \ref{St_limit_measure}. Ниже мы рассмотрим, как будут распределены ошибки приближения к решению, если мы будем итерировать алгоритм. Пусть $x_n$ --- $n$-ое фиксированное приближение к решению уравнения $ax=b$. Точное решение $x$ мы можем представить как сумму $x_n$ и поправки: $x=x_n+\Delta_n$. Подставляя ее в исходное уравнение, мы получаем уравнение относительно поправки $\Delta_n$:
\begin{equation}\label{equation_for_delta_n}
    a\Delta_n=b-ax_n.
\end{equation}
Пусть целое число $l_n$ такое, что
\begin{equation}\label{l_n}
    \frac{1}{2^{l_n+1}}<|b-ax_n|\leqslant\frac{1}{2^{l_n}}.
\end{equation}
Вместо уравнения \eqref{equation_for_delta_n} будем решать на QA уравнение
\begin{equation}\label{eq_delta_n_tilde}
    a\widetilde{\Delta}_n=2^{l_n}(b-ax_n)
\end{equation}
с неизвестной $\widetilde{\Delta}_n$. По предположению имеем, что 
\begin{equation}\label{Delta_n_tilde_distr_norm}
    \widetilde{\Delta}_n\sim\mathcal{N}\left(\frac{2^{l_n}(b-ax_n)}{a},\frac{1}{2a^2\beta^2 }\right).
\end{equation}
Так как $\Delta_n=\frac{1}{2^{l_n}}\widetilde{\Delta}_n$, то по свойствам нормального распределения получаем, что $\Delta_n\sim\mathcal{N}\left(\frac ba -x_n,\frac{1}{2^{2l_n}}\cdot\frac{1}{2a^2\beta^2 }\right)$.
Пусть $\xi_n\sim\mathcal{N}(0,1)$, тогда
\begin{equation*}
    \Delta_n\stackrel{d}{=}
    \frac ba -x_n+\frac{\xi_n}{2^{l_n}\sqrt{2}a\beta}.
\end{equation*}
Следующее приближение к решению будем вычислять по формуле
\[
    x_{n+1}=x_n+\Delta_n.
\]

Заметим, что следующая поправка $\Delta_{n+1}$ и число $l_{n+1}$ зависят от предыдущей поправки $\Delta_n$. Следующее предложение позволяет построить последовательность приближений $x_n$, учитывая эти зависимости.

\begin{proposition}\label{St_sequence_x_n_normal}
    Пусть $a,b\in\mathbb{R}\backslash\{0\}$, $\sigma>0$. Пусть $\xi_n\sim\mathcal{N}(0,1)$, $n\geqslant0$ --- независимые в совокупности случайные величины. Построим последовательности случайных величин $x_n$, $x_n'$ $l_n$, $l_n'$, $\Delta_n$, $\Delta_n'$ при $n\geqslant0$. Положим $x_0=x_0'=l_0=l_0'=0$. Дадим дальнейшие определения при $n\geqslant0$. Пусть условное распределение случайной величины $\Delta_n$ при условии $x_n$ равно $\mathcal{N}\left(\frac ba -x_n,\frac{\sigma^2}{2^{2l_n}}\right)$. Положим
    \begin{equation*}
        \Delta_n'=\frac ba -x_n'+\frac{\sigma\xi_n}{2^{l_n'}},
    \end{equation*}
    \begin{equation*}
        x_{n+1}=x_n+\Delta_n,
    \end{equation*}
    \begin{equation*}\label{x_n_reccurent}
        x_{n+1}'=x_n'+\Delta_n'.
    \end{equation*}
    Определим $l_{n+1}$ и $l_{n+1}'$ следующим образом: $l_{n+1},l_{n+1}'\in\mathbb{Z}$ и
    \[
        2^{l_{n+1}}|b-ax_{n+1}|,2^{l_{n+1}'}|b-ax_{n+1}'|\in\left(1/2,1\right].
    \] 
    
    Тогда $\left(\Delta_0,\Delta_1,\ldots,\Delta_n\right)\stackrel{d}{=}\left(\Delta_0',\Delta_1',\ldots,\Delta_n'\right)$ для любого $n\geqslant0$.
\end{proposition}

\begin{proof}
Доказательство будем вести индукцией по $n$. При $n=0$ мы имеем $\Delta_0\sim\mathcal{N}\left(\frac ba,\sigma^2\right)$ и  $\Delta_0'=\frac ba+\sigma\xi_0\sim\mathcal{N}\left(\frac ba,\sigma^2\right)$. Значит, $\Delta_0\stackrel{d}{=}\Delta_0'$.

Пусть $\left(\Delta_0,\ldots,\Delta_{n-1}\right)\stackrel{d}{=}\left(\Delta_0',\ldots,\Delta_{n-1}'\right)$, то есть для любого борелевского множества $A\subset\mathbb{R}^{n}$ выполнено $P\left(\left(\Delta_0,\ldots,\Delta_{n-1}\right)\in A\right)=P\left(\left(\Delta_0',\ldots,\Delta_{n-1}'\right)\in A\right)$. Покажем, что тогда для любого борелевского множества $A\subset\mathbb{R}^{n+1}$ будет выполнено  $P\left(\left(\Delta_0,\ldots,\Delta_{n}\right)\in A\right)=P\left(\left(\Delta_0',\ldots,\Delta_{n}'\right)\in A\right)$. Достаточно доказать, что
\[
    P\left(\left(\Delta_0,\ldots,\Delta_{n}\right)\in A_0\times\ldots\times A_{n}\right)=P\left(\left(\Delta_0',\ldots,\Delta_{n}'\right)\in A_0\times\ldots\times A_{n}\right)
\]
для любых борелевских $A_i\subset\mathbb{R}$, $i=0,1,\ldots,n$. По формуле полной вероятности получаем, что вероятность $P\left(\left(\Delta_0',\ldots,\Delta_{n}'\right)\in A_0\times\ldots\times A_{n}\right)$ равна
\[
    \int_{\mathbb{R}^n}
    P\left(\left(\Delta_0',\ldots,\Delta_{n}'\right)\in A_0\times\ldots\times A_{n}\,|\,\left(\Delta_0',\ldots,\Delta_{n-1}'\right)=(r_0,\ldots, r_{n-1})\right)dP'_{n-1},
\]
где $P'_{n-1}$ --- распределение вектора $\left(\Delta_0',\ldots,\Delta_{n-1}'\right)$ и интегрирование ведется по переменным $r_0,\ldots,r_{n-1}$. Если $(r_0,\ldots, r_{n-1})\notin A_0\times\ldots\times A_{n-1}$, то $\left(\Delta_0',\ldots,\Delta_{n}'\right)\notin A_0\times\ldots\times A_{n}$, и вероятность под знаком интеграла равна нулю. Поэтому последний интеграл равен
\begin{equation}\label{integral}
    \int\limits_{A_0\times\ldots\times A_{n-1}}
    P\left(\Delta_{n}'\in A_{n} \,|\,\Delta_0'=r_0,\ldots,\Delta_{n-1}'=r_{n-1}\right)dP'_{n-1}.
\end{equation}

Зафиксируем числа $r_0,\ldots,r_{n-1}$. Положим $S=r_0+\ldots+r_{n-1}$ и возьмем целое $l$ таким, что $2^l|b-aS|\in\left(\frac12,1\right]$. Покажем, что условное распределение случайной величины $\Delta_n'$ при условии $\Delta_0'=r_0,\ldots,\Delta_{n-1}'=r_{n-1}$ равно условному распределению случайной величины $\Delta_n$ при условии $\Delta_0=r_0,\ldots,\Delta_{n-1}=r_{n-1}$. 

При этих условиях $x_n=S$ и $x_n'=S$. Тогда по определению $l_n'$ получим, что $l_n'=l$. По определению $\Delta_n'$ мы имеем
\[
    \Delta_n'
    =\frac ba -S+\frac{\sigma\xi_n}{2^{l}}
    \sim\mathcal{N}\left(\frac ba-S,\frac{\sigma^2}{2^{2l}}\right),
\]
что есть условное распределение случайной величины $\Delta_n$ при условии $x_n=S$. 

По индукционному предположению мы имеем $P'_{n-1}=P_{n-1}$, где $P_{n-1}$ --- распределение вектора $\left(\Delta_0,\ldots,\Delta_{n-1}\right)$. Значит, интеграл \eqref{integral} равен
\[
    \int\limits_{A_0\times\ldots\times A_{n-1}}
    P\left(\Delta_{n}\in A_{n} \,|\,\Delta_0=r_0,\ldots,\Delta_{n-1}=r_{n-1}\right)dP_{n-1}
\]
\[
    =\int\limits_{\mathbb{R}^n}
    P\left(\left(\Delta_0,\ldots,\Delta_{n}\right)\in A_0\times\ldots\times A_{n}\,|\,\left(\Delta_0,\ldots,\Delta_{n-1}\right)=(r_0,\ldots, r_{n-1})\right)dP_{n-1}.
\]
Последний интеграл равен вероятности $P\left(\left(\Delta_0,\ldots,\Delta_{n}\right)\in A_0\times\ldots\times A_{n}\right)$, что и требовалось показать.
\end{proof}

Полагая, что $\sigma^2=\frac{1}{2a^2\beta^2}$ в предложении \ref{St_sequence_x_n_normal}, мы получаем, что $x_{n+1}'=\frac ba+\frac{\xi_n}{2^{l_n'}\sqrt{2}a\beta}$. Следующая теорема показывает, что эта последовательность при определенных условиях сходится к решению уравнения $ax=b$.

\begin{theorem}\label{Th_xn_normal_convergence_general}
	Пусть $a, b\in\mathbb{R}$, $a\neq0$, $\beta>0$, $\gamma$ --- постоянная Эйлера-Маскерони. Зафиксируем последовательность независимых в совокупности случайных величин $\xi_n\sim\mathcal{N}\left(0,1\right)$, $n\geqslant0$. Построим последовательности случайных величин $l_n$ и $x_n$ по правилу $l_0=0$, $x_0=0$ и при $n\geqslant0$ положим $x_{n+1}=\frac ba+\frac{\xi_n}{2^{l_n}\sqrt{2}a\beta}$, целое число $l_{n+1}$ таково, что $2^{l_{n+1}}|b-ax_{n+1}|\in\left(\frac12,1\right]$. Тогда
	\begin{enumerate}
	    \item[1)] если $s\in\left[1,\beta e^{\gamma/2}\right)$, то $s^n\left(x_n-\frac{b}{a}\right)\xrightarrow[n\rightarrow\infty]{\text{п.н.}}0$,
	    \item[2)] если $s>2\beta e^{\gamma/2}$, то $s^n\left(x_n-\frac{b}{a}\right)\xrightarrow[n\rightarrow\infty]{\text{п.н.}}\infty$,
		\item[3)] если $\beta<\frac12 e^{-\gamma/2}$, то $x_n\xrightarrow[n\rightarrow\infty]{\text{п.н.}}\infty$.
	\end{enumerate}
\end{theorem}

\begin{remark}
    Процесс, описанный в теореме соответствует процессу работы квантового компьютера. Процесс последовательного улучшения решения описывался в статьях \cite{RogersFloatingpoint, seismic, MML2022,Elliptic2023}, однако теоретические вопросы сходимости подобных итеративных алгоритмов не исследовались. 

    Из первого пункта теоремы следует, что если $\beta>e^{-\gamma/2}\approx\frac34$, то последовательность $x_n$ будет сходиться к решению уравнения $ax=b$ почти наверное и притом с экспоненциальной скоростью. Второй пункт устанавливает верхнюю границу скорости сходимости. Третий пункт устанавливает достаточное условие расходимости последовательности $x_n$: если $\beta$ мало (точность работы QA слишком плоха), то последовательность $x_n$ будет расходиться.
\end{remark}

Для доказательства теоремы \ref{Th_xn_normal_convergence_general} нам понадобится следующая лемма, непосредственно следующая из усиленного закона больших чисел Колмогорова \cite{Shiryaev}.

\begin{lemma}\label{Lemma_product_convergence}
    Пусть $\delta>0$ и пусть $\left(X_i\right)_{i=1}^{\infty}$ --- независимые в совокупности случайные величины такие, что для любого натурального $i$ существуют $\mathbb{E}\ln{|X_i|}$ и $\mathbb{E}\ln^2{|X_i|}$. Пусть дисперсии $\mathrm{Var}\left(\ln{|X_i|}\right)$ ограничены в совокупности. Тогда
    \begin{enumerate}
        \item[1)] если $\mathbb{E}\ln{|X_i|}<-\delta$ для любого $i$, то $X_1X_2\cdot\ldots\cdot X_n\xrightarrow[n\rightarrow\infty]{\text{п.н.}}0$,
        \item[2)] если $\mathbb{E}\ln{|X_i|}>\delta$ для любого $i$, то $X_1X_2\cdot\ldots\cdot X_n\xrightarrow[n\rightarrow\infty]{\text{п.н.}}\infty$.
    \end{enumerate}
\end{lemma}

\begin{proof}[Доказательство теоремы \ref{Th_xn_normal_convergence_general}]

Проведем предварительные построения. Рассмотрим случайную величину $z_n=\frac ba-x_n$. Оценим сверху $|z_{n+1}|$, используя определение случайных величин $x_n$ и $l_n$:
\begin{equation*}
    |z_{n+1}|
    =\left|\frac{\xi_n}{2^{l_n}\sqrt{2}a\beta}\right|
    =\left|\frac{(b-ax_n)\xi_n}{2^{l_n}(b-ax_n)\sqrt{2}a\beta}\right|
    <\left|\frac{\sqrt{2}(b-ax_n)\xi_n}{a\beta}\right|
    =\frac{\sqrt{2}}{\beta}|z_n\xi_n|.
\end{equation*}
Так как $|z_1|=\frac{|\xi_0|}{\sqrt{2}a\beta}$, то
\begin{equation}\label{estimate_z_n_up}
    |z_{n+1}|<\left|\frac{1}{2a}\right|\left(\frac{\sqrt{2}}{\beta}\right)^{n+1}|\xi_{0}\ldots\xi_n|.
\end{equation}
Аналогично получаем оценку $|z_{n+1}|$ снизу:
\begin{equation}\label{estimate_z_n_down}
    |z_{n+1}|
    \geqslant\left|\frac 1a\right|\left(\frac{1}{\sqrt{2}\beta}\right)^{n+1}|\xi_{0}\ldots\xi_{n}|.
\end{equation}
Найдем матожидание случайной величины $\ln|\xi_0|$:
\begin{equation*}
    \mathbb{E}\ln|\xi_0|
    =\frac{1}{\sqrt{2\pi}}\int_{-\infty}^{\infty}e^{-\frac{x^2}{2}}\ln{|x|}\,dx
    =\frac{2}{\sqrt{\pi}}\int_{0}^{\infty}e^{-x^2}\ln{\left(\sqrt{2}x\right)}\,dx 
    =\ln{\frac{1}{\sqrt{2}e^{\gamma/2}}},
\end{equation*}
где мы воспользовались соотношением для $\gamma$ \cite{EulerConstant}:
\begin{equation*}
    \int_{0}^{\infty}e^{-x^2}\ln{x}\,dx
    =-\frac{\sqrt{\pi}}4\left(\gamma+\ln{4}\right).
\end{equation*}
Тогда 
\begin{equation}\label{expectations_xis}
    \mathbb{E}\ln\frac{\sqrt{2}|\xi_0|}{\beta}
    =\ln{\frac{1}{\beta e^{\gamma/2}}},~~
    \mathbb{E}\ln\frac{|\xi_0|}{\sqrt{2}\beta}
    =\ln{\frac{1}{2\beta e^{\gamma/2}}}.
\end{equation}

Докажем пункт 1. Достаточно показать, что $\left(\frac{s\sqrt{2}}{\beta}\right)^n|\xi_{0}\ldots\xi_{n-1}|\xrightarrow[n\rightarrow\infty]{\text{п.н.}}0$ исходя из неравенства \eqref{estimate_z_n_up}. Ввиду \eqref{expectations_xis} и того, что $s\in\left[1,\beta e^{\gamma/2}\right)$, имеем
\[
    \mathbb{E}\ln\frac{s\sqrt{2}|\xi_0|}{\beta}
    =\ln{\frac{1}{\beta e^{\gamma/2}}}+\ln{s}<0.
\]
Значит, по лемме \ref{Lemma_product_convergence} произведение $\left(\frac{s\sqrt{2}}{\beta}\right)^n|\xi_{0}\ldots\xi_{n-1}|$ сходится к нулю почти наверное.

Докажем пункт 2. Достаточно показать, что $\left(\frac{s}{\sqrt{2}\beta}\right)^n|\xi_{0}\ldots\xi_{n-1}|\xrightarrow[n\rightarrow\infty]{\text{п.н.}}\infty$ исходя из неравенства \eqref{estimate_z_n_down}. Ввиду \eqref{expectations_xis} и того, что $s>2\beta e^{\gamma/2}$, имеем
\[
    \mathbb{E}\ln\frac{s|\xi_0|}{\sqrt{2}\beta}
    =\ln{\frac{1}{2\beta e^{\gamma/2}}}+\ln{s}>0.
\]
Значит, по лемме \ref{Lemma_product_convergence} произведение $\left(\frac{s}{\sqrt{2}\beta}\right)^n|\xi_{0}\ldots\xi_{n-1}|$ сходится к бесконечности почти наверное.

Докажем пункт 3. Достаточно показать, что $\left(\frac{1}{\sqrt{2}\beta}\right)^n|\xi_{0}\ldots\xi_{n-1}|\xrightarrow[n\rightarrow\infty]{\text{п.н.}}\infty$ исходя из неравенства \eqref{estimate_z_n_down}. Ввиду \eqref{expectations_xis} и того, что $\beta<\frac12 e^{-\gamma/2}$, имеем $\mathbb{E}\ln\frac{|\xi_0|}{\sqrt{2}\beta}=\ln{\frac{1}{2\beta e^{\gamma/2}}}>0$. Значит, по лемме \ref{Lemma_product_convergence} произведение $\left(\frac{1}{\sqrt{2}\beta}\right)^n|\xi_{0}\ldots\xi_{n-1}|$ сходится к бесконечности почти наверное.
\end{proof}


\subsection{Общий подход к исследованию сходимости алгоритмов решения линейного уравнения.}
\label{subsection_improving_general}

В этом разделе мы рассмотрим общий подход к построению и исследованию сходимости алгоритма, решающего уравнение $ax=b$. Мы рассмотрим общую схему последовательных приближений и докажем теорему \ref{Th_general_model}, позволяющую оценить скорость сходимости подобных алгоритмов. Домножая $a$ и $b$ на одинаковую степень двойки, можно добиться выполнения неравенства
\begin{equation}\label{eqa12}
    1/2 \leqslant a < 1.    
\end{equation}
В дальнейшем в статье мы будем предполагать, что выполнено неравенство \eqref{eqa12}.

Вернемся к построению последовательности приближений к решению уравнения $ax=b$. Имея очередное фиксированное приближение $x_n$, мы решаем на QA уравнение \eqref{eq_delta_n_tilde} относительно $\widetilde{\Delta}_n$, где целое $l_n$ выбирается в соответствии с \eqref{l_n}. Следующее приближение вычисляется как $x_{n+1}=x_n+2^{-l_n}\widetilde{\Delta}_n$. Распределение случайной величины $\widetilde{\Delta}_n$ зависит от выбранного нами алгоритма и определяется значениями $c_n$, $\sign(b-ax_n)$, $a$, $\beta$, где
\begin{equation}\label{eq_c_n}
    c_n=\frac{1}{2^{l_n}|b-ax_n|}.
\end{equation}
Заметим, что $c_n\in[1,2)$. В этом разделе мы не фиксируем конкретный алгоритм и соответствующее распределение $\widetilde{\Delta}_n$. В дальнейшем мы будем предполагать, что зависимость от $\sign(b-ax_n)$ имеет специальный вид, а именно существует такая функция $q(u, c, a, \beta)$, что если $\eta$ --- случайная величина, равномерно распределенная на $[0, 1]$, то
\begin{equation}\label{eq_q_st}
    \widetilde{\Delta}_n\,|\,x_n \stackrel{d}{=} \sign(b-ax_n) q(\eta, c_n, a, \beta)\,|\,x_n.
\end{equation}
Условие \eqref{eq_q_st} означает, что распределение поправок $\widetilde{\Delta}_n$ в случае положительных и отрицательных невязок отличается лишь знаком.

Заметим, что если функция $q(\cdot,c_n,a,\beta)$ равна обратной функции распределения закона ${\cal N}\left(\frac{1}{ac_n},\frac{1}{2a^2\beta^2}\right)$, то $q(\eta,c_n,a,\beta)\sim{\cal N}\left(\frac{1}{ac_n},\frac{1}{2a^2\beta^2}\right)$, и поправка $\widetilde{\Delta}_n$ распределена как в \eqref{Delta_n_tilde_distr_norm}.

По аналогии с предложением \ref{St_sequence_x_n_normal} верно следующее.

\begin{proposition}\label{prop_new_cut_norm_sequence_vsh}
    Пусть $b\neq0$, $\beta>0$. Пусть $\eta_n$, $n\geqslant0$ --- независимые в совокупности случайные величины, равномерно распределенные на промежутке $[0, 1]$. Зафиксируем функцию $q(\eta, c, a, \beta)$, определяемую выбранным алгоритмом последовательных приближений. 
    
    Введем функции 
    \[
        G_1(u_0)=q(u_0,1/b,a,\beta)-\frac ba,
    \]
    \begin{equation}
        G_{n+1}(u_0,u_1,\ldots,u_{n})
        =G_n(u_0,u_1,\ldots,u_{n-1})
        \left(1-a c_n' q(u_n,c_n',a,\beta)\right),~n\geqslant1, \label{G_n_general}
    \end{equation}
    где $u_i\in[0,1]$, $c_n'=\frac{1}{2^{l_n'}\left|aG_n\right|}$ и целое $l_n'$ выбрано так, что $2^{l_n'}|aG_n|\in\left(\frac12,1\right]$.

    Пусть $x_0=l_0=0$. Пусть $\widetilde{\Delta}_0\stackrel{d}{=}q(\eta_0,1/b,a,\beta)$ и при $n\geqslant1$ выполнено \eqref{eq_q_st} для $\eta=\eta_n$, где $l_n$ и $c_n$ определяются в \eqref{l_n} и \eqref{eq_c_n}. При $n\geqslant0$ положим
    \[
        x_{n+1}=x_n+2^{-l_n}\widetilde{\Delta}_n,
    \]
    \begin{equation}\label{x_n_new_general}
	x'_{n+1}=\frac ba+G_{n+1}(\eta_0,\ldots,\eta_n).
    \end{equation}

    Тогда $(x_1,\ldots,x_n)\stackrel{d}{=}(x_1',\ldots,x_n')$ для любого $n\geqslant1$.
\end{proposition}

\bigskip

Здесь и далее мы будем рассматривать последовательность $x_n$, заданную равенством \eqref{x_n_new_general}. Последовательность $x_n$ определяется выбором функции $q$. В следующей теореме мы будем использовать обозначения и определения из предложения \ref{prop_new_cut_norm_sequence_vsh}.

\begin{theorem}
\label{Th_general_model}
    Обозначим 
    \begin{equation}\label{r_th}
        r(u, a, \beta) = \max_{c \in [1, 2]}|1-c \cdot a \cdot q(u, c, a, \beta)|,~u\in[0,1].
    \end{equation}
    Пусть $E(a, \beta)$ --- математическое ожидание случайной величины $\ln r(\eta, a, \beta)$, где $\eta$ --- случайная величина, равномерно распределенная на $[0,1]$:
    $$
        E(a, \beta) = \int_0^1 \ln r(u, a, \beta)du.
    $$
    
    Тогда 
    \begin{enumerate}
        \item[1)] если $E(a, \beta) < 0$, то  $x_n\xrightarrow[n\to\infty]{\text{п.н.}}\frac ba$,
        \item[2)] если $\ln s+E(a, \beta)<0$ , то $s^n\left(x_n-\frac ba\right) \xrightarrow[n\to\infty]{\text{п.н.}}0$.
    \end{enumerate}
\end{theorem}

\begin{proof}
    Докажем первый пункт. Исходя из \eqref{x_n_new_general}, достаточно показать, что $G_{n+1}(\eta_0,\ldots,\eta_n)\xrightarrow[]{\text{п.н.}}0$. Используя \eqref{G_n_general}, получаем
    \begin{equation*}
        \left|G_{n+1}(\eta_0,\ldots,\eta_n)\right|
        =\left|G_1(\eta_0)\right|\prod\limits_{i=1}^n \left|1-a\,c_i\, q(\eta_i,c_i,a,\beta)\right|
        \leqslant
        \left|G_1(\eta_0)\right|\prod\limits_{i=1}^n r(\eta_i,a,\beta).
    \end{equation*}
    Так как математические ожидания случайных величин $\ln r(\eta_i, a, \beta)$ меньше нуля, то по лемме \ref{Lemma_product_convergence} получаем требуемое.

    Докажем второй пункт. Достаточно показать, что $s^{n+1}G_{n+1}(\eta_0,\ldots,\eta_n)\xrightarrow[]{\text{п.н.}}0$. Аналогично получаем
    \begin{equation*}
        s^{n+1}\left|G_{n+1}(\eta_0,\ldots,\eta_n)\right|
        \leqslant
        s\left|G_1(\eta_0)\right|\prod\limits_{i=1}^n \left(s\,r(\eta_i,a,\beta)\right).
    \end{equation*}
    Из неравенства $\mathbb{E}\ln(s\,r(\eta_i,a,\beta))<0$ следует требуемое.
\end{proof}

\subsection{Модели вычислений, основанные на усеченном нормальном распределении}
\label{subsection_improving_solution_cut_normal_distr}
\indent

В разделе \ref{section_preliminary} мы рассматривали представление \eqref{expression_for_x} переменной $x$ по положительным и отрицательным степеням двойки. В текущем разделе мы рассмотрим более ``гибкое'' представление:
\begin{equation}\label{expression_for_x_2}
    x=(d_2-d_1)\sum_{i=1}^{r}q_i2^{-i}+d_1,
\end{equation}
где $q_i\in\{0,1\}$, $d_1<d_2$, $r\in\mathbb{N}$. В таком представлении $x$ принимает значения в промежутке $[d_1,d_2)$. Коэффициенты $Q_{ij}$ в модели QUBO \eqref{qubo} находятся из подстановки представления \eqref{expression_for_x_2} в функцию \eqref{hamiltonian_expression}.

Предположим, что количество кубитов в QA стремится к бесконечности, то есть $r\rightarrow\infty$. Посмотрим, как при этом ведет себя распределение Больцмана на множестве, состоящем из чисел вида \eqref{expression_for_x_2}, с целевой функцией $(ax-b)^2$. Для начала введем следующее определение.

\begin{definition}\label{def_cut_normal}
    Пусть $\sigma>0$, $\mu,d_1,d_2\in\mathbb{R}$, $d_1<d_2$. Усеченное нормальное распределение $\mathcal{N}(\mu,\sigma^2,d_1,d_2)$ --- это распределение с функцией плотности
    \[
        f(t)\propto
        e^{-\frac{(t-\mu)^2}{2\sigma^2}}
        \mathbf{1}_{(d_1,d_2)}(t).
    \] 
    Если $d_1>d_2$, то обозначение $\mathcal{N}(\mu,\sigma^2,d_1,d_2)$ будет пониматься как усеченное нормальное распределение $\mathcal{N}(\mu,\sigma^2,d_2,d_1)$, а обозначение $(d_1,d_2)$ будет пониматься как интервал $(d_2,d_1)$.
\end{definition}

\begin{proposition}\label{St_limit_measure_cut_normal}
    Пусть $\beta>0$, $a,b\in\mathbb{R}$, $a\neq0$, $d_1<d_2$, $r\in\mathbb{N}$. Тогда
    \begin{equation*}
        \mathrm{B}\left(
        \beta,
        \left\{(d_2-d_1)\sum_{i=1}^{r}q_i2^{-i}+d_1\,:\,q_i\in\{0,1\}\right\},
        (ax-b)^2\right)
        \xrightarrow[r\rightarrow\infty]{d}
        \mathcal{N}\left(\frac{b}{a},\frac{1}{2a^2\beta^2},d_1,d_2\right).
    \end{equation*}
\end{proposition}

Таким образом, мы можем использовать усеченное нормальное распределение в качестве приближения к распределению Больцмана. 

Из формул \eqref{l_n},\eqref{eq_delta_n_tilde},\eqref{eqa12} следует, что 
\begin{equation}\label{eq_delta_12_st}
    |\widetilde{\Delta}_n| \leqslant 2, \quad \sign(b-ax_n) \widetilde{\Delta}_n \in [1/2, 2].
\end{equation}

Мы будем рассматривать несколько различных алгоритмов поиска поправки. На каждом шаге мы будем искать решение уравнения \eqref{eq_delta_n_tilde} на промежутке $\sign(b-ax_n)[d_{1}, d_{2}]$, представляя $\widetilde{\Delta}_n$ по аналогии с формулой \eqref{expression_for_x_2}:
$$
    \widetilde{\Delta}_n = \sign(b-ax_n)\left( (d_{2}-d_{1})\sum_{i=1}^{r}q_i2^{-i}+d_{1}\right).
$$
Cогласно предложению \ref{St_limit_measure_cut_normal} будем считать, что
\begin{equation*}\label{eq_trancn_st}
    \widetilde{\Delta}_n\sim \sign(b-ax_n)\mathcal{N}\left(\frac{1}{ac_n},\frac{1}{2a^2\beta^2 },d_1, d_2 \right).
\end{equation*}

Обозначим через $q(\eta, c, a, \beta)$ такие функции, что если $c, a, \beta$ фиксированы и $\eta$ равномерно распределена на промежутке $[0, 1]$, то
$$
    q(\eta, c, a, \beta) \sim \mathcal{N}\left(\frac{1}{ac}, \frac{1}{2a^2\beta^2}, d_1, d_2\right). 
$$
Функция $q$ представляет собой обратную функцию распределения соответствующего закона распределения и в явном виде записывается следующим образом:
\[
    q(u,c,a,\beta)
    =\frac{1}{ac}+
    \frac{1}{a\beta}\erf^{-1}
    \left((1-u)\,\erf\left(d_1 a\beta+\frac{\beta}{c}\right)
    +u\,\erf\left(d_2a\beta-\frac{\beta}{c}\right)\right),
\]
где $\erf(x)=\frac{2}{\sqrt{\pi}}\int_{0}^{x}e^{-t^2}dt$ --- функция ошибок.

Выбор алгоритма определяется выбором чисел $d_1,d_2$. Мы рассмотрим следующие алгоритмы:

\begin{description}
    \item[\normalfont{Алгоритм 1:}] $d_1 = -2$, $d_2 = 2$. Не учитывает знак поправки, а лишь наибольшее значение модуля $|\widetilde{\Delta}_n|$.
    \item[\normalfont{Алгоритм 2:}] $d_1 = 0$, $d_2 = 2$. Учитывает знак поправки и наибольшее значение модуля $|\widetilde{\Delta}_n|$.
    \item[\normalfont{Алгоритм 3:}] $d_1 = 1/2$, $d_2 = 2$. Учитывает знак поправки и наибольшее и наименьшее значение модуля $|\widetilde{\Delta}_n|$.
    \item[\normalfont{Алгоритм 4:}] $d_1 = 1/2$, $d_2 = 1$. Консервативный алгоритм, при котором гарантированно выполняется неравенство $r(u, a, \beta) \leq 1$, где $r$ определено в \eqref{r_th}. По теореме \ref{Th_general_model} такой алгоритм сходится для любых $a,\beta$. Отметим, что при этом точное значение $\frac{1}{ac_n}$ поправки $\widetilde{\Delta}_n$ не всегда лежит в интервале $[d_1, d_2]$.
\end{description}

На рис. \ref{fig:max_E_st} приведены сравнительные графики для соответствующих функций $E_{max}(\beta) = \max\limits_{a\in[1/2,1]} E(a,\beta)$, дающие пессимистичную оценку на скорость сходимости алгоритма. Если $E_{max}(\beta)<0$, то алгоритм сходится при любых $a \in [0.5, 1)$,  $b \in \mathbb{R}$, чем меньше значение $E_{max}(\beta)$, тем сходимость быстрее.

\begin{figure}[!h]
    \centering
    \includegraphics[width=0.69\textwidth]{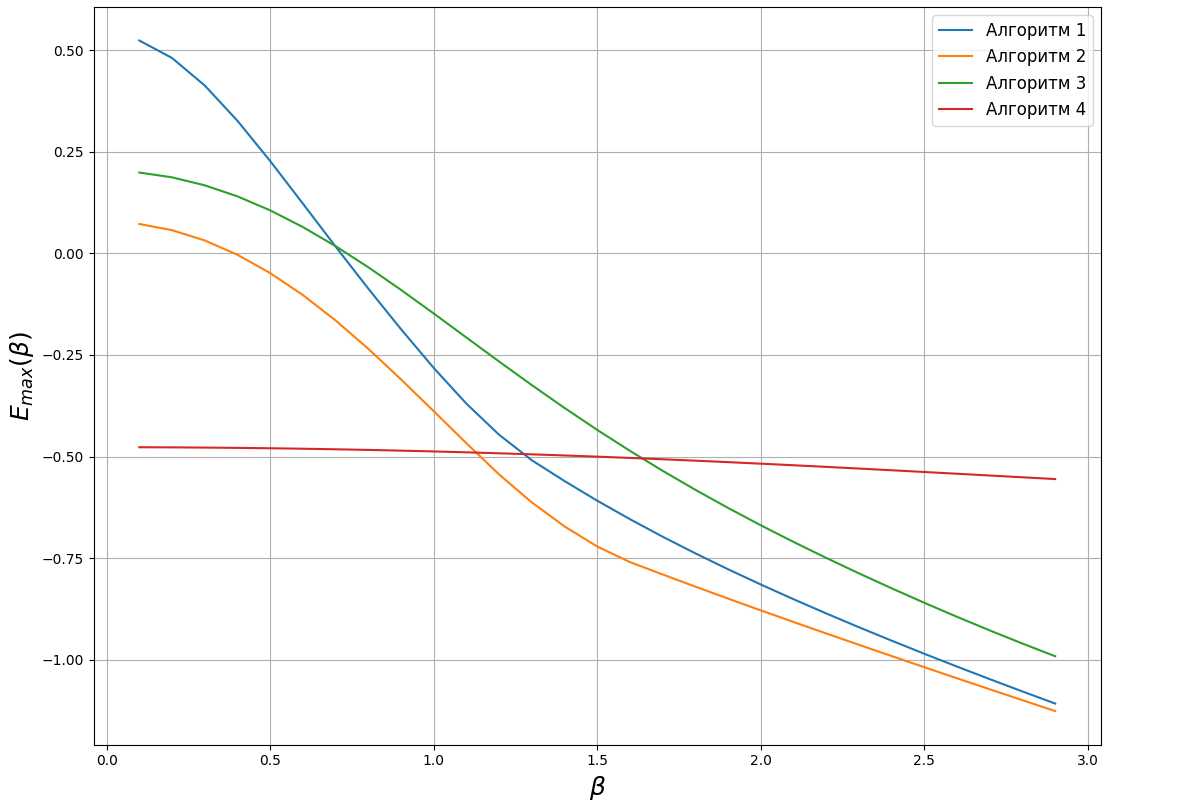}
    \caption{Графики функций $E_{max}(\beta)$ для различных алгоритмов. }
    \label{fig:max_E_st}
\end{figure}

Поскольку Алгоритм 4 при любом выборе $\widetilde{\Delta}_n$ уменьшает значение невязки, то он сходится в любом случае, что отражено в отрицательности функции $E_{max}(\beta)$ для всех значений $\beta$. При этом, поскольку включение $\frac{1}{ac} \in [0.5, 1]$ не всегда выполнено, алгоритм показывает не быструю скорость сходимости даже при больших значениях $\beta$. Алгоритмы 1-3 ведут себя приблизительно одинаково при больших значениях $\beta$, это объясняется высокой вероятностью получить значение $\widetilde{\Delta}_n$, близкое к точному решению уравнения \eqref{eq_delta_n_tilde}. Лучшие показатели сходимости наблюдаются у Алгоритма 2, учитывающего знак, но разрешающего малые значения поправки на каждом шаге. Его преимущество над Алгоритмом 3 вероятно объясняется уменьшением веса хвоста усеченного нормального распределения при котором $r(u, a, \beta) > 1$.

\subsection{Модели вычислений, основанные на распределении Больцмана}
\label{subsection_improving_solution_Boltzmann_distr}

В разделах \ref{subsection_improving_solution_normal_distr}, \ref{subsection_improving_solution_cut_normal_distr} мы рассматривали непрерывные распределения, приближающие распределение Больцмана. В этом разделе мы непосредственно рассмотрим модель вычислений, основанную на распределении Больцмана, и будем считать, что количество кубитов в QA конечно.

Мы будем рассматривать несколько различных алгоритмов поиска поправки $\widetilde{\Delta}_n$. В одной группе алгоритмов мы не будем учитывать правильный знак поправки и будем искать решение уравнения \eqref{eq_delta_n_tilde}, представляя $\widetilde{\Delta}_n$ как в формуле \eqref{expression_for_x}:
$$
    \widetilde{\Delta}_n = \vartheta q_{p} +\sum\limits_{i=r}^{p-1} 2^i q_i,
$$
где $r,p\in\mathbb{Z}$, $r<p$, $\vartheta=-2^p+2^r$.

В другой группе алгоритмов мы будем учитывать знак поправки и будем искать $\widetilde{\Delta}_n$, представляя его как
$$
    \widetilde{\Delta}_n = \mathrm{sign}(b-ax_n)\sum\limits_{i=r}^{p-1} 2^i q_i.
$$
Заметим, что количество $n_q$ участвующих в представлении $\widetilde{\Delta}_n$ кубитов в группе алгоритмов, не учитывающих знак, равно $p-r+1$, а в группе, учитывающих знак, равно $p-r$. Обозначим 
\[
    \Omega_{r,p}^{\pm}=\left\{\pm\sum_{i=r}^{p-1} q_i 2^i~:~q_i\in\{0,1\}\right\},~~
    \Omega_{r,p}^+=\Omega_{r,p}\cap[0,\infty).
\]

Выбор алгоритма определяется выбором $\Omega_{r,p}^{\pm}$ или $\Omega_{r,p}^+$ в качестве множества поиска поправки. Так как выполнено \eqref{eq_delta_12_st}, то достаточно брать $p\leqslant1$. 

Распределение поправки на $n$-ом шаге задается соотношением
\[
    \widetilde{\Delta}_n\sim\mathrm{sign}(b-ax_n)\mathrm{B}\left(\beta,\Omega_{r,p},\left(ax-\frac{1}{c_n}\right)^2\right),
\]
где $\Omega_{r,p}$ равно либо $\Omega_{r,p}^{\pm}$, либо $\Omega_{r,p}^+$.

Обозначим через $q_{r,p}(\eta, c, a, \beta)$, такие функции, что если $c, a, \beta$ фиксированы и $\eta$ равномерно распределена на промежутке $[0, 1]$, то
$$
    q_{r,p}(\eta, c, a, \beta) \sim \mathrm{B}\left(\beta,\Omega_{r,p},\left(ax-\frac1c\right)^2\right), 
$$
где $\Omega_{r,p}$ определяется в соответствии с выбранным алгоритмом. Функции $q_{r,p}$ представляют собой обратные функции распределения соответствующих законов распределения и определяются как
\[
    q_{r,p}(u, c, a, \beta)=\inf\left\{t\,|\,F_{r,p}(t)\geqslant u\right\},~u\in(0,1],
\]
где $F_{r,p}(t)$ --- функция распределения закона $\mathrm{B}\left(\beta,\Omega_{r,p},\left(ax-\frac1c\right)^2\right)$. Отметим, что функция $q_{r,p}(\cdot, c, a, \beta)$ кусочно-постоянная. Дополнительно определим 
$$
    q_{r,p}(0, c, a, \beta)=\lim\limits_{u\to0+}q_{r,p}(u, c, a, \beta).
$$

На рис. \ref{fig:max_E_Boltz} приведены сравнительные графики для соответствующих функций $E_{max}(\beta) = \max\limits_{a\in[1/2,1]} E(a,\beta)$, дающие пессимистичную оценку на скорость сходимости алгоритма. Если $E_{max}(\beta)<0$, то алгоритм сходится при любых $a \in [0.5, 1)$,  $b \in \mathbb{R}$, чем меньше значение $E_{max}(\beta)$, тем сходимость быстрее.

\begin{figure}[!h]
    \centering
    \includegraphics[width=0.80\textwidth]{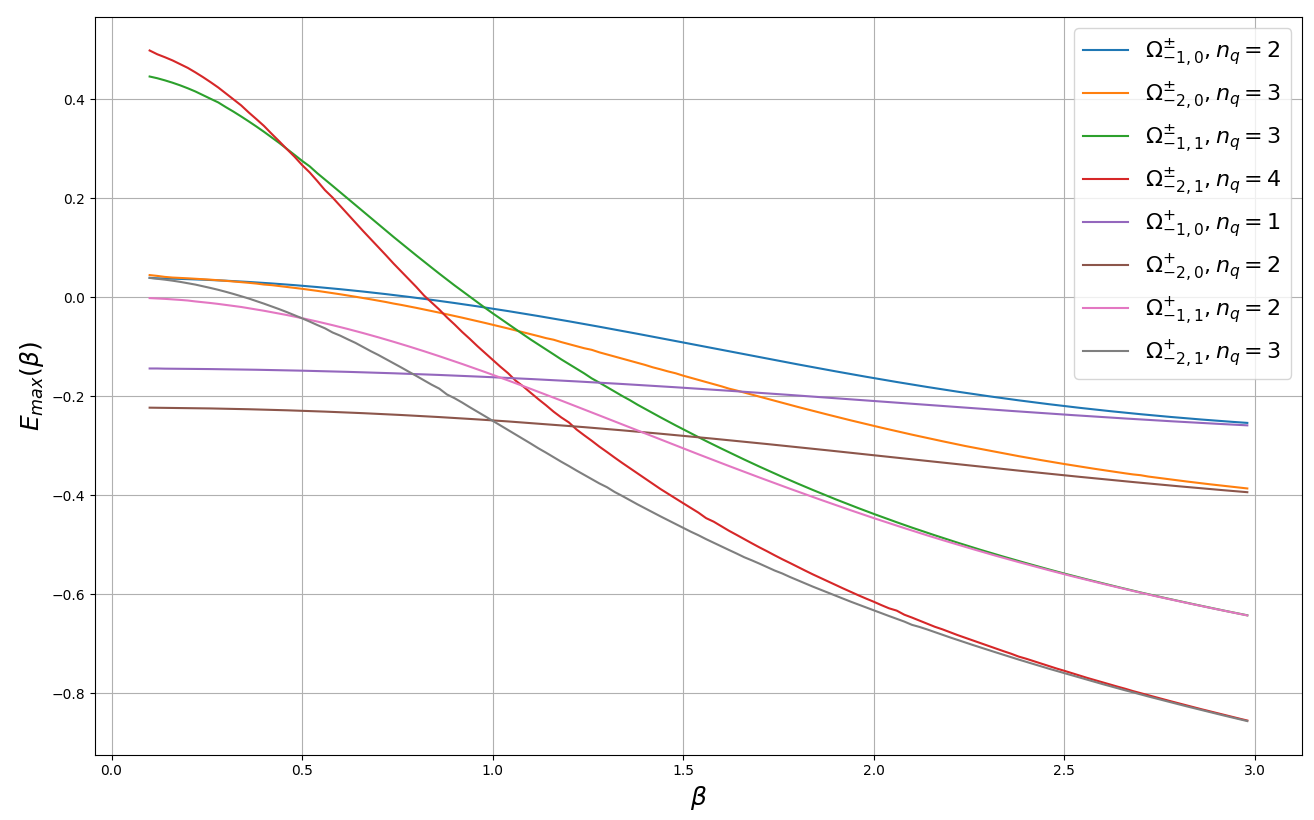}
    \caption{Графики функций $E_{max}(\beta)$ для различных $\Omega_{r,p}$. Число $n_q$ --- количество кубитов, кодирующих поправку.}
    \label{fig:max_E_Boltz}
\end{figure}

Из графиков видно, что при достаточно больших $\beta$ сходятся все методы, включая основанные на одном кубите. Методы, учитывающие знак поправки, в которых $\Omega_{r,p}=\Omega_{r, p}^+$, сходятся быстрее, чем не учитывающие знак поправки, в которых $\Omega_{r,p}=\Omega_{r, p}^{\pm}$. При этом методы, учитывающие знак поправки используют меньшее количество кубитов. Ожидаемо, с увеличением количества используемых кубитов скорость сходимости возрастает, но остается ниже чем предельная скорость, соответствующая усеченному нормальному распределению на рис. \ref{fig:max_E_st}. По аналогии со сравнением Алгоритмов 2 и 3 из раздела \ref{subsection_improving_solution_cut_normal_distr} отметим, что методы с $p = 1$ включают больше значений поправок, при которых точность может ухудшиться, но при этом гарантированно содержат наилучшую возможную поправку, в то время как методы с $p = 0$ гарантированно не ухудшают точность приближения на каждом шаге, но при этом имеют меньшую вероятность для оптимальной поправки. Как и в случае бесконечного количества кубитов, при больших $\beta$ методы с $p= 1$ оказываются более эффективными, чем при $p=0$.


\section{Выводы}
\indent

В статье рассмотрены несколько адаптивных итеративных методов для поиска корня линейного уравнения $ax = b$ при помощи устройства, работающего по принципу квантового отжига. Результат работы QA моделируется распределением Больцмана. Для широкого класса алгоритмов предложен метод доказательства их сходимости и оценки скорости сходимости. Рассмотрены алгоритмы с бесконечным количеством кубитов и с малым количеством кубитов. Показано, что при достаточно малом шуме скорость сходимости экспоненциальная. При этом алгоритмы, учитывающие знак поправки сходятся быстрее, чем не учитывающие знак.

\end{document}